\documentclass[reqno, a4paper, 11pt]{amsart}

\usepackage[T1]{fontenc}
\usepackage[utf8]{inputenc}
\usepackage{microtype}
\usepackage{mathtools}
\usepackage{amssymb} 
\usepackage{enumitem} 
\usepackage[foot]{amsaddr}
\usepackage{geometry}
\usepackage{epigraph}
\usepackage{tikz}

\usepackage{natbib}
\setcitestyle{authoryear,open={(},close={)}}

\usepackage[colorlinks=true,
            linkcolor=blue,
            citecolor=blue,
            urlcolor=blue]{hyperref}

\hypersetup{
    pdftitle={Perfect Secrecy in the Wild},
    pdfauthor={Costas Cavounidis, Massimiliano Furlan, Alkis Georgiadis‑Harris},
    pdfkeywords={Perfect Secrecy, Information Design}
}

\theoremstyle{definition}
\newtheorem{definition}{Definition}

\newtheorem{theorem}{Theorem}
\DeclareMathOperator\supp{supp}

\title[Perfect Secrecy in the Wild]{
    Perfect Secrecy in the Wild \\[3pt]
    \small{
        A Characterization
    }
}

\author{
    Costas Cavounidis
    \qquad Massimiliano Furlan
    \qquad Alkis Georgiadis‑Harris \\ \\  
    University of Warwick
}
\date{\today}

\begin{document}

\begin{abstract}
Alice wishes to reveal the state $X$ to Bob, if he knows some other information $Y$ also known to her. If Bob does not, she wishes to reveal nothing about $X$ at all. When can Alice accomplish this? We provide a simple necessary and sufficient condition on the joint distribution of $X$ and $Y$. Shannon's result on the perfect secrecy of the one-time pad follows as a special case.
\end{abstract}

\thanks{We are grateful to Fran\c{c}oise Forges,  Carlo Perroni
and Debraj Ray for useful conversations. Nevertheless, we absolve them of responsibility for any errors and omissions in this paper.}

\maketitle

\begin{epigraphs}
    \qitem{
        Of all the Knights in Appledore\\
        The wisest was Sir Thomas Tom.\\
        He multiplied as far as four,\\
        And knew what nine was taken from\\
        To make eleven. He could write\\
        A letter to another Knight.
    }%
    {A. A. Milne}
\end{epigraphs}

\section{Introduction}
We study a setting in which an agent---Alice---sends a public signal $Z$ to another---Bob. Alice wishes to inform Bob of the state $X$ if he knows the value of the private signal $Y$, also known to Alice.\footnote{Throughout, we assume that $X$ and $Y$ are finitely supported.} If he does not, she wishes that her signal gives no information about $X$ at all---a requirement \cite{shannon1949communication} terms \textit{perfect secrecy}. Equivalent to this is the case where Bob is known to know $Y$ and must surely learn $X$, but Alice's signal will be overheard by that dastardly eavesdropper, Eve, who should learn nothing about $X$.\footnote{Readers familiar with cryptographic terms may be inclined to call $x$ the ``plaintext'' or ``message'', $y$ the ``key'', and $z$ the ``ciphertext''.}

We find that Alice's problem has a solution \textit{if and only if} the matrix of conditional probabilities with elements $P_{Y\mid X}(y_j|x_i)$ is column-substochastic. Under this condition, we construct a plausible joint distribution of the state, private signal, and a public signal $Z$, such that observing any pair $(y,z)$ determines $x$, but conditioning on any $z$ alone says nothing about $X$. Conversely, when our condition fails, Alice's task is impossible.

The intuition for the necessity of our condition is as follows. Public signal realizations that reveal $x$ when combined with private signal $y$ must occur with zero probability given $(x^\prime,y)$ for every $x^\prime\neq x$. To maintain perfect secrecy, each signal realization taken on its own must maintain the prior marginal over $X$. Thus, signals revealing $x$ when combined with $y$ must have total probability at least $P_{XY}(x,y)/P_X(x)$. Summing over $x \in \supp X$ delivers our condition. 

The sufficiency of our condition is demonstrated by construction. We decompose the matrix of conditional probabilities into a convex combination of truncated permutation matrices using the Birkhoff-von Neumann Theorem. We then show that, reweighted by the prior marginal on $X$, the decomposition can be interpreted as a joint distribution on $X,Y$ and a public signal $Z$ satisfying our desiderata.

\cite{shannon1949communication} demonstrates that if the private signal is distributed uniformly and independently of the state, Alice's problem is solvable if and only if $\#\supp X \leq \# \supp Y$.\footnote{This is one of several celebrated results known as Shannon's One-Time Pad Theorem or Shannon's Perfect Secrecy Theorem.} Independence and uniformity of the private signal imply that each column of the conditional matrix sums to $(\#\supp X) /(\#\supp Y )$; thus our condition reduces, in this case, to Shannon's.

Related work examines what can be said with plausible deniability \citep{antic2025subversive}, and what can be transmitted while obfuscating other facts \citep{strack2024privacy}. Recent applied work studies screening data buyers with different information offerings \citep{zhao2024tailoring}.

Private communication is a fundamental building block of mechanism design, Bayes-correlated equilibrium, and a variety of related tools.\footnote{Conversely, public communication is a common way to introduce coordination by way of jointly controlled lotteries \citep{aumann1995repeated}. Our work shows that public signals can easily be repurposed for private communication, complicating standard results.} Understanding the foundations of such communication clarifies---and in our view expands---the domain of application of such methods.

\section{The Theorem}
We are given two finitely-supported random variables $X$ and $Y$, with joint distribution $P_{XY}$. Alice's problem is to produce a random variable $Z$ and a joint distribution $Q_{XYZ}$ satisfying three desiderata:

\begin{definition}
We say $Q_{XYZ}$ satisfies \textbf{consistency} if $Q_{XY}=P_{XY}$.
\end{definition}

Consistency requires that the joint distribution $Q_{XYZ}$ cohere with the given prior. This is necessary to interpret $Z$ as a public signal which Bob and Eve will use to update their beliefs. 

\begin{definition}
We say $Q_{XYZ}$ satisfies \textbf{informativeness} if $Q_{X \mid YZ}(\, \cdot \mid y,z)$ is degenerate whenever $Q_{YZ}(y,z)>0$.
\end{definition}

Informativeness requires that Bob is able to deduce $X$ from his observations of $Y$ and $Z$. 

\begin{definition}
We say $Q_{XYZ}$ satisfies \textbf{perfect secrecy} if $Q_{X \mid Z}=P_X$.
\end{definition}

Perfect secrecy requires that $X$ and $Z$ are independent, so that the public signal leaks no information to Eve. Since Shannon, it is the first-best standard for secure communication.

The statement of our theorem follows. We should point out that clever measure theory cannot improve on this---the restriction on $Z$ being finitely supported is merely a matter of expedience.

\begin{theorem}\label{thm:main}
There exists a random variable $Z$ and a finitely-supported joint distribution \(Q_{XYZ}\) satisfying consistency, informativeness, and perfect secrecy if and only if 
\begin{equation*}\label{eq:column_condition}
    \sum_{x \in \supp X} P_{Y \mid  X}(y \mid x) \leq 1
    \quad 
    \text{ for all } 
    y \in \supp Y.
\end{equation*}
\end{theorem}


\begin{proof}

{\bf(Necessity)}. Suppose there exists a random variable \(Z\) and a finitely-supported joint distribution \(Q_{XYZ}\) that satisfies 
     \(Q_{XY} = P_{XY}\) (consistency),
     \(Q_{X\mid YZ}\) degenerate (informativeness), and
     \(Q_{X\mid Z}=P_{X}\) (perfect secrecy).
Let 
\(
    \phi(x, y) = \{z \in\supp Z : Q_{XYZ}(x,y,z) > 0\}
\); 
it is the set of public signals that are sent with positive probability when \(X=x\) and \(Y=y\). Informativeness implies that for any fixed \((y,z) \in \supp Y \times \supp Z\), if \(Q_{XYZ}(x,y,z) > 0\) for some \(x \in \supp X\), then \(Q_{XYZ}(x^{\prime},y,z) = 0\) for all \(x^{\prime} \neq x\). Hence, for any two different \(x,x^{\prime} \in \supp X\), and any $y \in \supp Y$, the sets $\phi(x,y)$ and $\phi(x^{\prime},y)$ are disjoint; \(\phi(x,y) \cap \phi(x^{\prime},y) = \varnothing\).

Note that for any \((x,y,z) \in \supp X \times \supp Y \times \supp Z\),
\begin{equation}\label{eqn:star}
    Q_{XYZ}(x,y,z) \leq \sum_{y^{\prime} \in \supp Y} Q_{XYZ}(x,y^{\prime},z) = Q_{XZ}(x,z).
\end{equation}

For any fixed \(y\in \supp Y\), we have
\[
\begin{aligned}
    \sum_{x\in\supp X}\frac{P_{XY}(x,y)}{P_{X}(x)}
        &= \sum_{x\in\supp X}  \frac{1}{P_{X}(x)} \sum_{z \in\supp Z}Q_{XYZ}(x,y,z)
        && \text{(consistency)}\\[5pt]
        &= \sum_{x\in\supp X} \frac{1}{P_{X}(x)} \sum_{z\in\phi(x,y)}
           Q_{XYZ}(x,y,z)
        && \text{(definition of }\phi)\\[5pt]
        &\leq \sum_{x\in\supp X} \frac{1}{P_{X}(x)} \sum_{z\in\phi(x,y)}
           Q_{XZ}(x,z)
        && (Q_{XYZ}\le Q_{XZ} \text{ from } \eqref{eqn:star})\\[5pt]
        &= \sum_{x\in\supp X} \frac{1}{P_{X}(x)} \sum_{z\in\phi(x,y)} P_{X}(x) Q_{Z}(z)
        && \text{(perfect secrecy)}\\[5pt]
        &= \sum_{x\in\supp X} \sum_{z\in\phi(x,y)} Q_{Z}(z)\leq 1 
        && ( \phi(x,y) \cap \phi(x^{\prime},y)= \varnothing ) . 
\end{aligned}
\] 

%
{\bf(Sufficiency)}. We may enumerate \(\supp X = \{x_1, \ldots, x_n\}\) and \(\supp Y = \{y_1,\ldots, y_m\}\).
Suppose, as postulated, that for all \(j \in \{1,\ldots,m\}\), \(\sum_{i=1}^n P_{Y \mid  X}(y_j \mid x_i) \leq 1\). 
Let \(\mathbf{P}_{Y \mid X}\) be the \(n \times m\) matrix of conditionals whose $ij$-th entry is \(P_{Y \mid  X}(y_j \mid x_i)\). 
The condition implies \(\mathbf{P}_{Y \mid X}\) is column-substochastic, with 
\begin{equation*}
   n = \sum_{i=1}^n \sum_{j=1}^m  P_{Y \mid  X}(y_j \mid x_i)  = \sum_{j=1}^m \sum_{i=1}^n P_{Y \mid  X}(y_j \mid x_i) \, \leq \,
   \sum_{j = 1}^m 1 = m .
\end{equation*}
We will now extend $\mathbf{P}_{Y \mid X}$ to a doubly-stochastic matrix. If \(n = m\), \(\mathbf{P}_{Y \mid X}\) is already doubly-stochastic. In that case let \(\mathbf{P}^{\mathsf{ext}}_{Y \mid X} := \mathbf{P}_{Y \mid X}\). If \(n \leq m\), let \(\mathbf{P}^{\mathsf{ext}}_{Y \mid X}\) be the \(m \times m\) doubly stochastic matrix obtained by extending \(\mathbf{P}_{Y \mid X}\) with \(m-n\) new identical rows,
\begin{equation*}
    \mathbf{r}^\top := \dfrac{\mathbf{1}_m^\top - \mathbf{1}_n^\top \mathbf{P}_{Y \mid X}}{m - n}  \geq 0
\end{equation*}
\begin{equation*}
    \mathbf{P}^{\mathsf{ext}}_{Y \mid X} := 
    \begin{pmatrix}
        \mathbf{P}_{Y \mid X} \\[5pt]
        \mathbf{1}_{(m-n)} \mathbf{r}^\top
    \end{pmatrix},
\end{equation*}
where \(\mathbf{1}_m\) (resp. \(\mathbf{1}_{(m-n)}\)) is the all-ones vector of dimension \(m\) (resp. \(m-n\)).

By the Birkhoff-von Neumann Theorem,\footnote{See \cite{birkhoff1946decomposition} and
\cite{vonneumann1953decomposition}.} \( \mathbf{P}^{\mathsf{ext}}_{Y \mid X} \) can be expressed as a convex combination of \(m \times m\) permutation matrices,
\begin{equation}\label{convex_combo}
    \mathbf{P}^{\mathsf{ext}}_{Y \mid X} = \sum_{k=1}^p \alpha_k 
    \mathbf{\Pi}^{(k)},
\end{equation}
where \(\alpha_k >0\), \(\sum_{k=1}^p \alpha_k = 1\), and each \(\mathbf{\Pi}^{(k)}\) is an \(m \times m\) permutation matrix.
Define a finitely-supported random variable \(Z\) by identifying its  realizations $\{z_1,\dots,z_p\}$ with the permutation matrices in (\ref{convex_combo}); and construct the joint distribution of all three variables as so:
\begin{equation}\label{eqn:construction}
    Q_{XYZ}(x_i, y_j, z_k) = \alpha_k P_X(x_i) \mathbf{\Pi}^{(k)}_{ij}.
\end{equation}
This is indeed a distribution since, 
\[
    \sum_{k=1}^p \alpha_k \sum_{i=1}^n P_X(x_i) \sum_{j=1}^m \mathbf{\Pi}^{(k)}_{ij} = \sum_{i=1}^n \sum_{j=1}^m P_X(x_i)P_{Y \mid X}(y_j \mid x_i) =1,
\]
and has marginal \(Q_Z(z_k) = \alpha_k\).
Moreover, it satisfies, consistency
    \[
        \begin{aligned}
        \sum_{k=1}^p Q_{XYZ}(x_i, y_j, z_k) 
        &=  P_X(x_i) \sum_{k=1}^p \alpha_k \mathbf{\Pi}^{(k)}_{ij}
        \quad && \text{(construction \eqref{eqn:construction})}\\[5pt]
        &= P_X(x_i) [\mathbf{P}_{Y \mid X}]_{ij}
        && \text{(decomposition of \(\mathbf{P}_{Y \mid X}\))}\\[5pt]
        &= P_X(x_i) P_{Y \mid X}(y_j \mid x_i)
        && \text{(definition of \(\mathbf{P}_{Y \mid X}\))}\\[5pt]
        &= P_{XY}(x_i,y_j) ; \\[5pt]
         \end{aligned}
    \]
perfect secrecy,
    \[
        \begin{aligned}
         Q_{XZ}(x_i,z_k)
         &= \sum_{j=1}^m Q_{XYZ}(x_i, y_j, z_k) \\[5pt]
         &=   \alpha_k P_X(x_i)  \sum_{j=1}^m \mathbf{\Pi}^{(k)}_{ij}
         \quad && \text{(construction \eqref{eqn:construction})}\\[5pt]
         &=   \alpha_k P_X(x_i)
         && \text{(\(\mathbf{\Pi}^{(k)}_{ij}=1 \Rightarrow \mathbf{\Pi}^{(k)}_{i^{\prime}j}=0\))}\\[5pt]
         &=   Q_Z(z_k) P_X(x_i)
         && \text{(\(Q_Z(z_k) = \alpha_k\))} ;
         \end{aligned}
    \]
and informativeness:
    for any \((y_j,z_k)\) pair, if \( Q_{XYZ}(x_i, y_j, z_k) > 0 \) then \(\mathbf{\Pi}^{(k)}_{ij}= 1\) implying \(\mathbf{\Pi}^{(k)}_{i^{\prime}j}= 0\) for any \(i^\prime \neq i\) (as $\mathbf{\Pi}^{(k)}$ is a permutation matrix) and therefore \( Q_{XYZ}(x_{i^{\prime}}, y_j, z_k) = 0 \). Hence, \(Q_{X \mid YZ}\) is degenerate. \end{proof}


Under Shannon's assumptions, our condition collapses to his: when $X\perp Y$ and $Y$ is uniform, all entries of the conditional matrix are identical. Thus, given it is row-stochastic, column-substochasticity means there are (weakly) more columns than rows; or equivalently, the support of $Y$ is more numerous than that of $X$. This recovers Shannon's condition. But his result goes further in one way---it shows that in his case, the mapping from pairs $(x,y)$ to public signals $z$ can be chosen to be deterministic. A natural question, then, is when this is possible in general.

In our set-up, to be able to choose $Q_{Z \mid XY}$ to be degenerate, it must be that $Q_{XYZ}(x,y,z)>0$ implies that $Q_{XYZ}(x,y,z^\prime)=0$ for all $z^\prime\neq 0$. Thus, the values of $P_{Y \mid X}$ on the support of each $Q_{XY \mid Z}(\, \cdot \mid z)$ must be identical. But this in turn implies that $P_{Y \mid X}$ is in fact a Latin rectangle, with potential symbol repetition.

We note that if we think of the construction problem as mapping each $(x,y)$ pair into a distribution over public signals, the solution does not depend on the marginal on $X$. Rather, it only depends on the conditional $P_{Y \mid X}$.

\section{Conclusion}
Our work asks when an agent can convey a message secretly to another, given the two possess side-information. We generalize Shannon's classic result by providing an easy-to-check substochasticity condition that is necessary and sufficient for secret communication via a public signal. The result clarifies the potential for and limits of secrecy in settings where only public channels are available. This is of fundamental importance to economic theorists, whose assumptions on private and public communication often drive their results.

\bibliographystyle{plainnat}
\bibliography{bibfile}

\end{document}